\documentclass[submission,copyright,creativecommons]{eptcs}
\providecommand{\event}{SOS 2007} 

\usepackage{iftex}
\usepackage{enumitem}
 \usepackage{graphicx}
 \usepackage{amsmath}
\usepackage{amsthm}
\newtheorem{theorem}{Theorem}[section]
\newtheorem{remark}{Remark}[section]
\newtheorem{example}{Example}[section]
\newtheorem{proposition}{Proposition}[section]
\newtheorem{lemma}{Lemma}[section]
\usepackage{float}
\theoremstyle{definition}
\newtheorem{definition}{Definition}
[section]
\newcommand{\noi}{\noindent}

\ifpdf
  \usepackage{underscore}         
  \usepackage[T1]{fontenc}        
\else
  \usepackage{breakurl}           
\fi

\title{On a Generalization of the Christoffel Tree: Epichristoffel Trees}
\author{Abhishek Krishnamoorthy
\institute{Madras Christian College\\ Chennai, Tamil Nadu, India}
\institute{Department of Mathematics}
\email{abhishek@mcc.edu.in}
\and
Robinson Thamburaj
\institute{Madras Christian College\\
Chennai, Tamil Nadu, India}
\institute{Department of Mathematics}
\email{robinson@mcc.edu.in}
\and
Durairaj Gnanaraj Thomas
\institute{Madras Christian College\\
Chennai, Tamil Nadu, India}
\institute{Department of Mathematics}
\email{dgthomasmcc@yahoo.com}
}
\def\titlerunning{On a Generalization of the Christoffel Tree: Epichristoffel Trees}
\def\authorrunning{Abhishek Krishnamoorthy, Robinson Thamburaj \& Durairaj Gnanaraj Thomas}
\def\copyrightholders{Abhishek Krishnamoorthy, Robinson Thamburaj\\ \& Durairaj Gnanaraj Thomas}
\begin{document}
\maketitle
\begin{abstract}
Sturmian words form a family of one-sided infinite words over a binary alphabet that are obtained as a discretization of a line with an irrational slope starting from the origin. A finite version of this class of words called Christoffel words has been extensively studied for their interesting properties. It is a class of words that has a geometric and an algebraic definition, making it an intriguing topic of study for many mathematicians. Recently, a generalization of Christoffel words for an alphabet with 3 letters or more, called epichristoffel words, using episturmian morphisms has been studied, and many of the properties of Christoffel words have been shown to carry over to epichristoffel words; however, many properties are not shared by them as well. In this paper, we introduce the notion of an epichristoffel tree, which proves to be a useful tool in determining a subclass of epichristoffel words that share an important property of Christoffel words, which is the ability to factorize an epichristoffel word as a product of smaller epichristoffel words. We also use the epichristoffel tree to present some interesting results that help to better understand epichristoffel words.  
\end{abstract}

\section{Introduction}

Sturmian sequences have appeared several times in history in the works of several mathematicians such as Bernoulli, Christoffel and Markov. They are sequences over a 2-letters alphabet that code discrete lines, due to which they appear in different fields of mathematics such as discrete geometry and number theory. They also have the property of being balanced. Christoffel words are the finite version of these sequences and have been studied extensively \cite{Onaninvolution,Onchristoffelclasses,Observatio,Combonwords,Sturmianmorphisms,Algebraic}. For a comprehensive understanding of Christoffel words, we refer the readers to \cite{Combonwords},\cite{Fromchris} and \cite{Balanceprop}. These are the smallest words with respect to the lexicographic order which are conjugate to a finite standard Sturmian word. 

Genevieve Paquin has studied a generalization of Christoffel words called epichristoffel words, in which episturmian morphisms are used to determine if a word belongs to an epichristoffel class \cite{Onagen}. Although epichristoffel words share many of the same properties as Christoffel words, Genevieve raised some open problems regarding epichristoffel words. These include the ability to characterize the epichristoffel word of each conjugacy class, whether epichristoffel words satisfy a type of balanced property, and whether there is an epichristoffel word of any length over a k-letter alphabet for a fixed $k\geq3$. In this paper, we have partially answered such questions by introducing infinite binary trees called epichristoffel trees motivated by the definition of the Christoffel tree. We note that there can be different epichristoffel trees based on the epichristoffel word that serves as the root of the tree. It has been shown in \cite{Quelquesmots} that every non-trivial Christoffel word can be factorized in a unique way into two words where each one of them is a Christoffel word. This is called the standard factorization of the word. The Christoffel tree is an infinite binary tree in which each Christoffel word appears exactly once in its standard factorization form. Although the epichristoffel trees that we introduce in this article do not contain every epichristoffel word over a $k$-letter alphabet and not all epichristoffel words have a factorization into smaller epichristoffel words, we show that the epichristoffel tree can be a useful tool in helping to determine the existence of epichristoffel words of various lengths. This tree exhibits subclasses of epichristoffel words that do or do not satisfy a factorization property, such as Christoffel words. Through Theorem 5.1, we have shown that this tree provides a characterization of the epichristoffel word of a conjugacy class. The study of epichristoffel words is interesting because they seem to relate to the Fraenkel conjecture which states that for a $k$-letter alphabet, there exists a unique infinite word up to letter permutation and conjugation that is balanced and has pair-wise distinct letter frequencies. Thus, a better understanding of epichristoffel words might help prove this conjecture.

\section{Definitions and Notations}

Throughout this paper, we use the notation $A$ to denote a finite ordered alphabet. A finite word is an element of the free monoid $A^\ast$ and if $w = w[0] w[1] \ldots w[n-1]$, with $w[i] \in A$ then $w$ is said to be a finite word of length $n$. Unless specified otherwise every word discussed in this paper will be a finite word. The notation $|w|$ is used to denote the length of the word $w$ and the notation $|w|_{a_i}$ is used to determine the number of occurrences of the letter $a_i$ in $w$, where $a_i \in A$. By convention, the empty word is denoted by $\lambda$ and its length is 0.

The conjugacy class $[w]$ of a finite word $w$ of length $n$ is the set of all words of the form $w[i] w[i+1] \ldots w[n-1] w[0] \ldots w[i-1]$, for $0 \le i \le n-1$. If two words $w_{1}$ and $w_{2}$ are conjugate to each other, we denote it by $w_{1}\equiv w_{2}$. If $w$ is primitive and the smallest word in its conjugacy class with respect to the lexicographic order, then $w$ is called a Lyndon word. 

A word $f$ is a factor of a word $w$ if $w = pfs$ for some $p, s \in A^\ast$. $f$ is called a prefix or a suffix respectively if $p$ or $s$ is the empty word.  We refer the readers to \cite{Onagen} and \cite{Combonwords} for most of the basic notations and definitions used in this paper. 
\section{Christoffel Words and Epichristoffel Words}
\begin{definition}\cite{Combonwords}
The lower Christoffel path of slope $\frac{a}{b}$ over the binary alphabet $A=\{x,y\}$, where $a$ and $b$ are relatively prime positive integers is the path in the plane from (0, 0) to $(b, a)$ in the integer lattice $Z \times Z$ that satisfies the following two conditions:
\begin{itemize}
\item The path lies below the line segment that begins at (0, 0) and ends at $(b, a)$.

\item The region in the plane enclosed by the path and the line segment contains no other part of $Z \times Z$ besides those of the path.
\end{itemize}
By encoding every horizontal step in the lower Christoffel path by the letter $x$ and every vertical step in the lower Christoffel path by the letter $y$ we get word of length $a+b$ over a binary alphabet called the Christoffel word of slope $\frac{a}{b}$.
\end{definition}
\vspace{-1em}
\begin{figure}[H]
\centering
\includegraphics[scale=0.5]{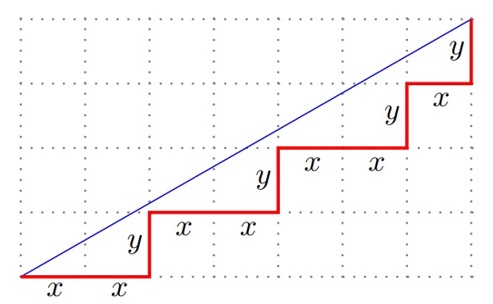}
\caption{The Christoffel word of slope $\frac{4}{7}$}
\label{fig1}
\end{figure}
\begin{lemma}\cite{Observatio}
A word $w$ is a Christoffel word if and only if $w$ is a balanced Lyndon word.
\end{lemma}

\begin{definition}\cite{Combonwords}
The label of a point $(i, j)$ on the lower Christoffel path of slope $\frac{a}{b}$ is the number $\frac{ia-jb}{b}$ which represents the vertical distance from the point $(i, j)$ to the line segment from (0, 0) to $(b, a)$.
\end{definition}
\begin{figure}[h!]
\centering
\includegraphics[scale=0.5]{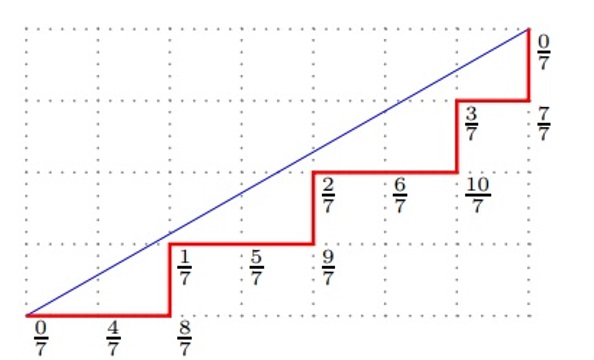}
\caption{The labels of the points on the Christoffel path of slope $\frac{4}{7}$}
\label{fig2}
\end{figure}
\begin{definition}\cite{Combonwords}
The standard factorization of the Christoffel word w of slope $\frac{a}{b}$ is the factorization $w = (w_1, w_2)$ where $w_1$ encodes the portion of the Christoffel path from (0, 0) to the closest point $C$ on the path having label $\frac{1}{b}$ and $w_2$ encodes the portion from $C$ to $(b, a)$. 
\end{definition}
\begin{definition}\cite{Episturmianwords}
For a finite alphabet $A$ and $a_1, a_2 \in A$, consider the following endomorphisms of $A^\ast$:
\begin{enumerate}[label=(\roman*)]
\item $\psi_{a_1}(a_1) = \overline{\psi}_{a_1}(a_1) = a_1$

\item $\psi_{a_1}(a_k) = a_1 a_k$ if $a_k \in A \setminus \{a_1\}$

\item $\overline{\psi}_{a_1}(a_k) = a_k a_1$ if $a_k \in A \setminus \{a_1\}$

\item $\theta_{a_1a_2}(a_1) = a_2$, $\theta_{a_1a_2}(a_2) = a_1$, $\theta_{a_1a_2}(a_k) = a_k$, $a_k \in A \setminus \{a_1, a_2\}$
\end{enumerate} 

\end{definition}
\begin{definition}\cite{Onagen}
    The set of episturmian morphisms is the monoid generated by the morphisms $\psi_{a_1}$,$\overline{\psi}_{a_1}$,\\$\theta_{a_1a_2}$ under composition. The set of pure episturmian morphisms is the submonoid generated by $\psi_{a_1}$ and $\overline{\psi}_{a_1}$.
\end{definition}
\begin{definition}\cite{Onagen}
A word $w \in A^\ast$ belongs to an epichristoffel class if it is the image of a letter by an episturmian morphism. 
\end{definition}

\begin{definition}\cite{Onagen}
A word $w \in A^\ast$ is epichristoffel if it is the unique Lyndon word occurring in an epichristoffel class.\ A word is called $c$-epichristoffel if it is conjugate to an epichristoffel word.\
\end{definition}
\begin{definition}\cite{Onagen}
    If $w$ is an epichristoffel word over the alphabet $A=\{a_{0},a_{1},...,a_{k-1}\}$ then the $k$-tuple $p = (p_0, p_1, \ldots, p_{k-1})$ is called an epichristoffel $k$-tuple if $p_{i}=|w|_{a_{i}}$ for $0\leq i\leq k-1$. 
\end{definition}
It has been shown in \cite{Combonwords} that for a given pair of positive integers $a, b$ there exists a Christoffel word having a number of $x'$s and $b$ number of $y'$s if and only if $a$ and $b$ are relatively prime. For a $k$-letter alphabet, $k \geq 3$, Paquin presented an algorithm to determine the existence of an epichristoffel $k$-tuple based on the following definition and results.
\begin{definition}\cite{Onagen}
Let $p = (p_0, p_1, \ldots, p_{k-1})$ be a $k$-tuple of non-negative integers. The operator $T : N^k \to Z^k$ is defined over the $k$-tuple $p$ as \newline
$T(p) = T(p_0, p_1, \ldots, p_{k-1}) = (p_0, p_1, \ldots, p_{i-1}, (p_i - \sum\limits_{j = 0, j \neq i}^{k-1} p_j), p_{i+1}, \ldots, p_{k-1})$, where $p_i \geq p_j$, $\forall$ $j \neq i$.
\end{definition}
\begin{proposition}\cite{Onagen}\label{prop3.1}
Let $p$ be a $k$-tuple. There exists an epichristoffel word with occurrence number of letters $p$ if and only if iterating $T$ over $p$ yields a $k$-tuple $p^\prime$ with $p_j^\prime = 0$ for $j \neq m$ and $p_m^\prime = 1$, for a unique $m$ such that $0 \le m \le k-1$.
\end{proposition}

\begin{example}
Consider the 3-tuple (2, 3, 7). Then, $T(2, 3, 7) = (2, 3, 2)$, $T^2(2, 3, 7) = T(2, 3, 2) = (2, -1, 2)$. Hence there exists no epichristoffel word corresponding to the 3-tuple (2, 3, 7).

On the other hand, there exists an epichristoffel word corresponding to the 3-tuple (1, 4, 2) since $T(1, 4, 2) = (1, 1, 2)$, $T^2(1, 4, 2) = T(1, 1, 2) = (1, 1, 0)$, $T^3(1, 4, 2) = T(1, 1, 0) = (1, 0, 0)$.
\end{example}
\begin{lemma}\cite{Onagen}
Let $w \in A^\ast$ be a $c$-epichristoffel word. Then, there exists a $c$-epichristoffel word $u \in A^\ast$, $|u| > 1$ and an episturmian morphism $\phi \in \{\psi_{a_0}, \overline{\psi}_{a_0}\}$, with $a_0 \in A$, such that $w = \phi(u)$ if and only if $|w|_{a_0} > |w|_{a_i}$ for all $a_i \in A$, $i \neq 0$.
\end{lemma}
We refer the readers to the iteration found in the proof of Lemma 3.2, as the algorithm used to determine the epichristoffel word is based on this iteration. Given below is an example of constructing the epichristoffel word based on the mentioned iteration.
\begin{example}
If $A = \{x, y, z\}$, then, for the 3-tuple (1, 4, 2) describing the occurrence numbers of $x, y$ and $z$ respectively, the construction of the epichristoffel word is as follows:\\
$(1, 4, 2) \to^y (1, 1, 2) \to^z (1, 1, 0) \to^y (1, 0, 0)$\\
$\psi_{y}\psi_{z}\psi_{y}(x) = \psi_{y}\psi_{z}(yx)  = \psi_y(zyzx) = yzyyzyx$\\
Since this is obtained by a standard episturmian morphism to a letter, this standard episturmian word is a representative of the epichristoffel conjugacy class. The word $yzyyzyx$ is thus a $c$-epichristoffel word and the smallest word with respect to the lexicographic order in its conjugacy class that is, $xyzyyzy$ is the epichristoffel word, assuming $x < y < z$. 
\end{example}
\begin{lemma}\cite{Somecombprop}
A Christoffel word can always be written as the product of two Christoffel words.
\end{lemma}

\begin{lemma}\cite{Onagen}\label{lem3.4}
An epichristoffel word cannot always be written as the product of two epichristoffel words.
\end{lemma}
\begin{example}
    Consider the epichristoffel word $xzyzzyz$ over the alphabet $A = \{x, y, z\}$. There are no factorizations of this word into epichristoffel words.
\end{example}

\begin{lemma}\cite{Onagen}\label{lem3.5}
Any $c$-epichristoffel word $w$ of length $n > 1$, can be non-uniquely written as the product of two $c$-epichristoffel words.
\end{lemma}
\begin{example}
    A factorization of the $c$-epichristoffel word $yzyyzyx$ is $yzy,yzyx$ obtained by considering the iteration in Example 3.2 in the following way:
    $\psi_{y}\psi_{z}\psi_{y}(x) = \psi_{y}\psi_{z}(yx)  = \psi_{y}\psi_{z}(y) ,\psi_{y}\psi_{z}(x)\newline=\psi_{y}(zy),\psi_{y}(zx)=yzy,yzyx$. \newline
    This factorization is used in the later part of the paper to construct an epichristoffel tree. 
\end{example}

\section{The Christoffel and Stern-Brocot trees}
\begin{definition}\cite{Combonwords}
The Christoffel tree is an infinite binary tree where each node is of the form $(u, v)$ which represents a Christoffel word occurring in its standard factorization form and the left, right descendants of which are $(u, uv)$ and $(uv, v)$ respectively. The root of this tree is the Christoffel word of slope $\frac{1}{1}$, i.e. $(x, y)$.
\end{definition}
\begin{lemma}\cite{Combonwords}
Every Christoffel word appears exactly once on the Christoffel tree. 
\end{lemma}
\begin{definition}\cite{Combonwords}
The mediant of two fractions $\frac{a}{b}$ and $\frac{c}{d}$, denoted by $\frac{a}{b} \oplus \frac{c}{d}$ is $\frac{a+c}{b+d}$. This operation gives rise to the Stern-Brocot sequence as follows:\\
Let $S_0$ denote the sequence $\frac{0}{1}$, $\frac{1}{0}$ (we view $\frac{1}{0}$ as a formal fraction for the purpose of the construction of the successive terms of the sequence). For $i > 0$, $S_i$ is constructed from $S_{i-1}$ by inserting between consecutive elements of the sequence their mediant. The first few iterations of this process yields the following sequence:\\
$\frac{0}{1}, \frac{\mathbf{1}}{\mathbf{1}}, \frac{1}{0}$ \hspace*{3mm} $(S_1)$\\
$\frac{0}{1}, \frac{\mathbf{1}}{\mathbf{2}}, \frac{1}{1}, \frac{\mathbf{2}}{\mathbf{1}}, \frac{1}{0}$ \hspace*{3mm} $(S_2)$\\
$\frac{0}{1}, \frac{\mathbf{1}}{\mathbf{3}}, \frac{1}{2}, \frac{\mathbf{2}}{\mathbf{3}}, \frac{1}{1}, \frac{\mathbf{3}}{\mathbf{2}}, \frac{2}{1}, \frac{\mathbf{3}}{\mathbf{1}}, \frac{1}{0}$ \hspace*{3mm} $(S_3)$\\
$\frac{0}{1}, \frac{\mathbf{1}}{\mathbf{4}}, \frac{1}{3}, \frac{\mathbf{2}}{\mathbf{5}}, \frac{1}{2}, \frac{\mathbf{3}}{\mathbf{5}}, \frac{2}{3}, \frac{\mathbf{3}}{\mathbf{4}}, \frac{1}{1}, \frac{\mathbf{4}}{\mathbf{3}}, \frac{3}{2}, \frac{\mathbf{5}}{\mathbf{3}}, \frac{2}{1}, \frac{\mathbf{5}}{\mathbf{2}}, \frac{3}{1}, \frac{\mathbf{4}}{\mathbf{1}}, \frac{1}{0}$ \hspace*{3mm} $(S_4)$ \newline
We have indicated the mediants obtained in each iteration in bold.
\end{definition}
\begin{definition}\cite{Combonwords}
The Stern-Brocot tree is an infinite binary tree in which the vertices of the $i^{th}$ $(i > 0)$ level are the mediants obtained in the $i^{th}$ iteration of the Stern-Brocot sequence. 
\end{definition}
\begin{figure}[H]
\centering
\includegraphics[scale=0.5]{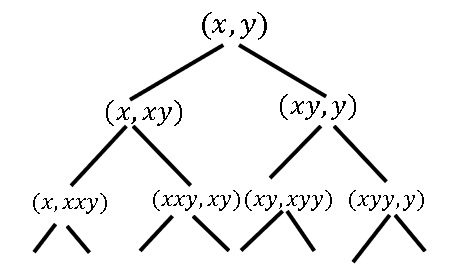}
\includegraphics[scale=0.5]{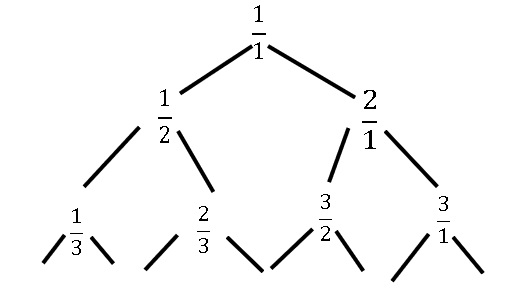}
\caption{The Christoffel and Stern-Brocot trees}
\label{fig3}
\end{figure}
\begin{theorem}\cite{Combonwords}
The Christoffel tree is isomorphic to the Stern-Brocot tree via the map that associates to the vertex $(u, v)$ of the Christoffel tree, the fraction $\frac{|uv|_y}{|uv|_x}$. The inverse map associates to a fraction $\frac{a}{b}$ the pair $(u, v)$ where $(u, v)$ is the standard factorization of the Christoffel word of slope $\frac{a}{b}$.
\end{theorem}
\begin{definition}\cite{Locterms}
For a natural number $k$, the $k^{th}$ left diagonal of the Stern-Brocot tree $L_k$ is the sequence made up of each $k^{th}$ term from each level beginning at the first level and the $k^{th}$ right diagonal of the Stern-Brocot tree $R_k$ is the sequence made up of each $k^{th}$ term taken from the end of each level beginning at the first level.

The first few left and right diagonals are mentioned below:\\
$L_1 = \left\{ \frac{1}{1}, \frac{1}{2}, \frac{1}{3}, \ldots \right\}$, $L_2 = \left\{ \frac{2}{1}, \frac{2}{3}, \frac{2}{5}, \ldots \right\}$, $L_3 = \left\{ \frac{3}{2}, \frac{3}{5}, \frac{3}{8}, \ldots \right\} \dots$\\
$R_1 = \left\{ \frac{1}{1}, \frac{2}{1}, \frac{3}{1}, \ldots \right\}$, $R_2 = \left\{ \frac{1}{2}, \frac{3}{2}, \frac{5}{2}, \ldots \right\}$, $R_3 = \left\{ \frac{2}{3}, \frac{5}{3}, \frac{8}{3}, \ldots \right) \dots$
\end{definition}
\noi{\bf Note:} If a particular level of the Stern-Brocot tree does not have a $k^{th}$ term, we omit that level and begin from the first level which has the $k^{th}$ term. For example, the left diagonal $L_{4}$ is obtained by taking the $4^{th}$ term from each row begining at the first level, but since the first and second levels of the tree do not have a $4^{th}$ term, we begin with the $4^{th}$ term of the $3^{rd}$ level.
\newline
\noi{\bf Notation:} If $L_k = \{L_{k,1}, L_{k,2}, L_{k,3}, \ldots\}$ and $L_m = \{L_{m,1}, L_{m,2}, L_{m,3}, \ldots\}$ then the sum of $L_k$ and $L_m$ is denoted by $L_k \oplus L_m$ and represents the left diagonal $\{L_{k,1} \oplus L_{m,1}, L_{k,2} \oplus L_{m,2}, \ldots\}$.
\begin{lemma}\cite{Locterms}\label{lem4.2}
If $k = 2^i(2j+1)$ where $i, j \in N$, then 
\begin{enumerate}[label=\textit{\roman*})]
\item \(
L_{2k} = \left\{
\begin{array}{ll}
L_k \oplus L_{1^+}, & \text{if } j = 0 \\
L_k \oplus L_{j+1}, & \text{if } j > 0
\end{array}
\right.
\quad \text{where } L_{1^+} = \left\{ \frac{1}{0}, \frac{1}{1}, \frac{1}{2}, \dots \right\}
\)

\item $L_{2k+1} = L_{k+1} \oplus L_{j+1}$.
\end{enumerate}
\end{lemma}
\begin{remark}
The goal of Lemma 4.2 is to show that any left diagonal in the Stern-Brocot tree is obtained by the mediant sum operation "$\oplus$" of some previous two left diagonals. For example, the left diagonal $L_6$ is obtained by $L_3\oplus{}L_2$. This is the motivation for the Stern-Brocot trees that we construct corresponding to our epichristoffel trees, which is shown in Theorem 5.2,as it follows the same rules of construction. The only difference is instead of adding two fractions with the mediant operation we add two tuples with the mediant operation. Therefore, the left diagonals of the Stern-Brocot trees that correspond to our epichristoffel trees also share the same property that any left diagonal can be obtained by the mediant sum of some previous two left diagonals.    
\end{remark}
\begin{theorem}\cite{Locterms}
If $\frac{a}{b}$ is the $k^{th}$ entry of the $n^{th}$ row of the Stern-Brocot tree where $n > 1$, then the $k^{th}$ entry of the $(n+1)th$ row is $\frac{a}{a+b}$.  
\end{theorem}

\begin{theorem}\cite{Locterms}
If $\frac{a}{b}$ is the $k^{th}$ entry of the $n^{th}$ row of the Stern-Brocot tree where $n > 1$, then the $k^{th}$ entry of the $(n+1)th$ row is $\frac{a+b}{b}$.  
\end{theorem}

From the Stern-Brocot tree we can thus deduce that if there exists a Christoffel word of slope $\frac{a}{b}$ then there will always exist a Christoffel word of slope $\frac{a}{b+na}$ and a Christoffel word of slope $\frac{a+nb}{b}$ for any natural number $n$. More results on the Stern-Brocot tree can be found in \cite{LinkingCalkin} and \cite{Locterms}.
\section{Epichristoffel Trees}
In this section, we introduce infinite binary trees for epichristoffel words similar to the Christoffel tree.  

Consider any $c$-epichristoffel word $w$ over the ordered alphabet $A=\{a_{0},a_{1},...,a_{k-1}\}$ of length $n > 1$. Since $w$ can be nonuniquely written as the product of two $c$-epichristoffel words, as shown in Example 3.4. Consider such a factorization of $(u, v)$ of $w$ where $u$ and $v$ are $c$-epichristoffel.

\begin{lemma}\label{lem5.1}
If $w = uv$ is a $c$-epichristoffel word whose factorization is obtained as explained above, then the words $w_1 = uuv$ and $w_2 = uvv$ are also $c$-epichristoffel words.
\end{lemma}
\begin{proof}
Since $w$ is a $c$-epichristoffel word, it is obtained by the application of a sequence of episturmian morphisms to some letter, say $a_{j} \in A$, where $0\leq j \leq k-1$. \newline
Therefore, \newline $w=\psi_{a_{i_{1}}}\psi_{a_{i_{2}}}...\psi_{a_{i_{l}}}(a_{j})$, where $0 \leq i_{1},i_{2},...,i_{l} \leq k-1$ and $i_{l} \neq j$. \newline
$=\psi_{a_{i_{1}}}\psi_{a_{i_{2}}}...\psi_{a_{i_{l-1}}}(a_{i_{l}}a_{j})$ \newline
$=\psi_{a_{i_{1}}}\psi_{a_{i_{2}}}...\psi_{a_{i_{l-1}}}(a_{i_{l}})\psi_{a_{i_{1}}}\psi_{a_{i_{2}}}...\psi_{a_{i_{l-1}}}(a_{j})$ \newline
Thus, $w=uv$, where $u=\psi_{a_{i_{1}}}\psi_{a_{i_{2}}}...\psi_{a_{i_{l-1}}}(a_{i_{l}})$ and $v=\psi_{a_{i_{1}}}\psi_{a_{i_{2}}}...\psi_{a_{i_{l-1}}}(a_{j})$
\newline
Consider now the word $w_{1}$ formed by the application of the following episturmian morphisms to the letter$a_{j}$: \newline
$w_{1}=\psi_{a_{i_{1}}}\psi_{a_{i_{2}}}...\psi_{a_{i_{l-1}}}\psi_{a_{j}}\psi_{a_{i_{l}}}(a_{j})$ \newline
$=\psi_{a_{i_{1}}}\psi_{a_{i_{2}}}...\psi_{a_{i_{l-1}}}\psi_{a_{j}}(a_{i_{l}}a_{j})$  \newline
$=\psi_{a_{i_{1}}}\psi_{a_{i_{2}}}...\psi_{a_{i_{l-1}}}\psi_{a_{j}}(a_{i_{l}})\psi_{a_{i_{1}}}\psi_{a_{i_{2}}}...\psi_{a_{i_{l-1}}}\psi_{a_{j}}(a_{j})$ \newline
$=\psi_{a_{i_{1}}}\psi_{a_{i_{2}}}...\psi_{a_{i_{l-1}}}(a_{j}a_{i_{l}})\psi_{a_{i_{1}}}\psi_{a_{i_{2}}}...\psi_{a_{i_{l-1}}}(a_{j})$ \newline
$=\psi_{a_{i_{1}}}\psi_{a_{i_{2}}}...\psi_{a_{i_{l-1}}}(a_{j})\psi_{a_{i_{1}}}\psi_{a_{i_{2}}}...\psi_{a_{i_{l-1}}}(a_{i_{l}})\psi_{a_{i_{1}}}\psi_{a_{i_{2}}}...\psi_{a_{i_{l-1}}}(a_{j})$ \newline
$=vuv$ \newline
$\equiv uvv$
\newline Therefore, $w_{2}$ is also $c$-epichristoffel.
\newline 
Similarly, it can be seen that $w_{1}$ is obtained by computing 
$\psi_{a_{i_{1}}}\psi_{a_{i_{2}}}...\psi_{a_{i_{l}}}\psi_{a_{i_{l}}}(a_{j})$ and is also \\$c$-epichristoffel.

\end{proof}
\subsection{The construction of an epichristoffel tree}
We construct the epichristoffel tree as follows: Let $w$ be an arbitrary epichristoffel word of length $n$ corresponding to the $k$-tuple $(x_1, x_2, \dots, x_k)$. The word $w$ is obtained by applying a sequence of episturmian morphisms on a letter $a_{j} \in A$, after which the smallest word in its conjugacy class is considered. In the process of doing so we factorize $w$ into two $c$-epichristoffel words say, $u$ and $v$ as shown in the proof of Lemma \ref{lem5.1}. \newline
Let $(l_1, l_2, \ldots, l_k)$ and $(m_1, m_2, \ldots, m_k)$ denote the corresponding $k$-tuples for the words $u$ and $v$ respectively. 

The epichristoffel word $w = w[1] w[2] \ldots w[n]$ is the smallest word in the conjugacy class obtained after the above episturmian morphims are applied. Due to this construction either $w[1] w[2] \ldots w[|u|]$ or $w[1] w[2] \ldots w[|v|]$ will be the epichristoffel word corresponding to $u$ or $v$. We use the word $w$ as the root of the tree. Suppose, without loss of generality that if $w[1] w[2] \ldots w[|v|]$ is the epichristoffel word corresponding to $v$, then the root of the tree is $(u^\prime, v^\prime)$ where\\
$u^\prime = w[1] w[2] \ldots w[|v|]$ and $v^\prime = w[|v|+1] w[|v|+2] \ldots w[|v|+|u|]$\\
The left and right descendants of this tree, as in the case of the Christoffel tree are $(u^\prime, u^\prime v^\prime)$ and $(u^\prime v^\prime, v^\prime)$ respectively. These in turn are once again epichristoffel due to Lemma \ref{lem5.1} and the fact that $u^\prime v^\prime$ is the smallest with respect to the lexicographic ordering. Suppose, if $w[1] w[2] \ldots w[|u|]$ is the epichristoffel word corresponding to $u$, then the root of the tree is $(u^\prime, v^\prime)$ where $u^\prime = w[1] w[2] \ldots w[|u|]$ and $v^\prime = w[|u|+1] w[|u|+2] \ldots w[|u|+|v|]$ This construction is illustrated in the example below:
\begin{example}\label{ex5.1}
Let $A = \{x, y, z\}$. Consider the 3-tuple (1,2,4). To find $u$ and $v$, we consider the episturmian morphisms applied to construct the word $zyzzyzx$.\\
$(1,2,4) \to^z (1,2,1) \to^y (1,0,1) \to^z (1,0,0)$\\
$\psi_{z}\psi_{y}\psi_{z}(x) = \psi_{z}\psi_{y}(zx) = \psi_{z}\psi_{y}(z) \psi_{z}\psi_{y}(x) = \psi_z(yz) \psi_z(yx) = zyz \ zyzx$\\
Thus, $u = zyz$ and $v = zyzx$, whose corresponding tuples are (0,1,2) and (1,1,2). The epichristoffel word corresponding to $zyzzyzx$ is $w = xzyzzyz$ and $w[1] w[2] w[3] w[4] = xzyz$ which is the epichristoffel word corresponding to the tuple (1,1,2).

The epichristoffel tree is obtained by considering $w$ as the root with the factorization $(xzyz,zyz)$ and the left and right descendants of each node follows the rule mentioned above.
\begin{figure}[h!]
\centering
\includegraphics[scale=0.5]{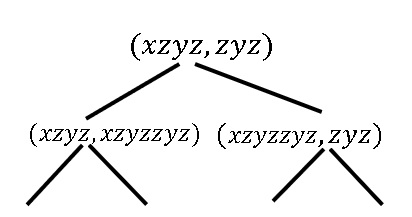}
\caption{The epichristoffel tree with the root word $(xzyz,zyz)$}
\label{fig5}
\end{figure}
\end{example}
Similar to the Stern-Brocot tree that is isomorphic to the Christoffel tree via the map that associates to each Christoffel word $w = (u,v)$ to the fraction $\frac{|w|_a}{|w|_b}$, we construct an infinite binary tree of $k$-tuples that can be associated to the epichristoffel tree via the map that associates each $c$-epichristoffel word $w = (u,v)$ to the $k$-tuple $(x_1, x_2, \ldots, x_k)$ where $x_i = |w|_{a_i}$.
\begin{definition}
The mediant of two $k$-tuples $(y_1, y_2, \ldots ,y_k)$ and $(z_1, z_2, \ldots ,z_k)$ is the be $k$-tuple\\
$(y_1+z_1, y_2+z_2, \ldots, y_k+z_k)$.\\
Let $(u_1^\prime, u_2^\prime, \ldots, u_k^\prime)$ and $(v_1^\prime, v_2^\prime, \ldots ,v_k^\prime)$ be two $k$-tuples corresponding to the words $u^\prime$ and $v^\prime$.\\
Define the sequence $S_i$ as follows:\\
$S_0 = (u_1^\prime, u_2^\prime, \ldots, u_k^\prime), (v_1^\prime, v_2^\prime, \ldots ,v_k^\prime)$ and\\
$S_i$ for $i > 1$ is obtained from $S_{i-1}$ by inserting between two consecutive elements of the sequence, their mediants. The mediants constructed in the $i$-$th$ iteration of the above process for $i > 0$ are the vertices in the $i$-$th$ level of this tree. We call this as the Stern-Brocot tree corresponding to the epichristoffel tree. 
\end{definition}
The Stern-Brocot sequence and tree corresponding to the epichristoffel tree of Example \ref{ex5.1} is shown below:
\begin{gather*}
S_0 = (1,1,2), (0,1,2)\\
S_1 = (1,1,2), (\mathbf{1},\mathbf{2},\mathbf{4}), (0,1,2)\\
S_2 = (1,1,,2), (\mathbf{2},\mathbf{3},\mathbf{6}), (1,2,4), (\mathbf{1},\mathbf{3},\mathbf{6}), (1,1,2)
\end{gather*}
and so on.
\vspace{-1em}
\begin{figure}[H]
\centering
\includegraphics[scale=0.5]{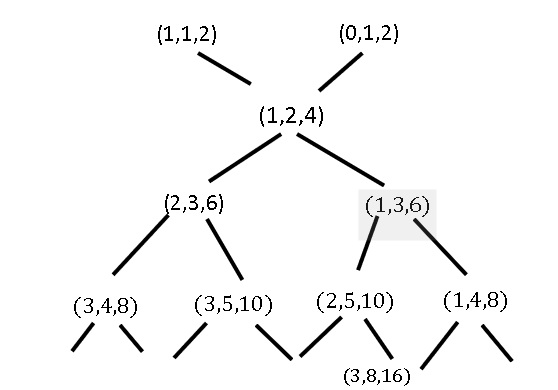}
\caption{The Stern-Brocot tree for Example \ref{ex5.1}}
\label{fig6}
\end{figure}
\begin{example}\label{ex5.2}
Let $A = \{x,y,z\}$. Consider the 3-tuple (3,2,1). The $c$-epichristoffel word corresponding to this tuple is $xy \ xyxz$. Here $u = xy$ and $v = xyxz$, whose corresponding tuples are (1,1,0) and (2,1,1) respectively. Hence, the epichristoffel word corresponding to this word is itself, that is $w = xyxyxz$. The respective epichristoffel and Stern-Brocot trees  are shown in Figures 7 and 8. 
\begin{figure}[h!]
\centering
\includegraphics[scale=0.5]{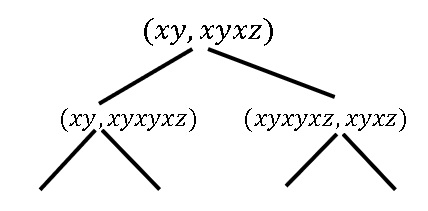}
\caption{The epichristoffel tree with the root word $(xy,xyxz)$}
\label{fig7}
\end{figure}

\begin{figure}[h!]
\centering
\includegraphics[scale=0.5]{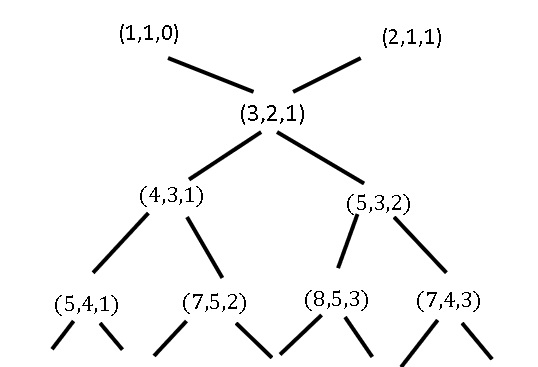}
\caption{The Stern-Brocot tree for Example \ref{ex5.2}}
\label{fig8}
\end{figure}
\end{example}
\begin{remark}
    To construct an epichristoffel tree we first begin with an epichristoffel word of length $n$, which serves as the root of the tree. To obtain an epichristoffel word we apply a set of episturmian morphisms to a letter. Once this is done, we may or may not have the epichristoffel word depending on whether the word is the smallest according to the lexicographic ordering in its conjugacy class. If the word obtained is not the smallest in its conjugacy class, we have obtained a c-epichristoffel word. This has been explained in Lemma 5.1. While applying the set of episturmian morphisms to the letter, we also factorize the word as shown in proof of Lemma 5.1 and write it as $(u,v)$, adhering to the notation of the construction of the Christoffel tree. If $w=w[1]w[2]\ldots w[n]$ is the smallest word in the conjugacy class of $(u,v)$, then either $w[1]w[2]\ldots w[|u|]$ or $w[1]w[2]\ldots w[|v|]$ is the epichristoffel word corresponding to the c-epichristoffel word u or the c-epichristoffel word $v$. This has been illustrated in Example 5.1 and Example 5.2. If without loss of generality, $w[1]w[2]\ldots w[|u|]$ is the epichristoffel word corresponding to the c-epichristoffel word $u$, then $w[u+1]w[u+2]\ldots w[|u+v|]$ is the epichristoffel word corresponding to the c-epichristoffel word $v$. We now begin the construction of our epichristoffel tree taking $(u',v')$ as the root of our tree where $u'=w[1]w[2]\ldots w[|u|]$  and $v'=w[u+1]w[u+2]\ldots w[|u+v|]$. The left and right descendants of this tree, as in the case of the Christoffel tree are $(u',u'v')$ and $(u'v',v')$ respectively. Since $(u',v')$ is an epichristoffel word, it is the smallest word with respect to the lexicographic order in its conjugacy class. Therefore, $u'<v'$ in the lexicographic ordering. From which we can conclude that the word $(u',u'v')$ is also smallest in its conjugacy class as we have only added the epichristoffel word u' that is smaller as a prefix. Similarly, we can deal with the case of  $(u'v',v')$ where we have added the epichristoffel word $v'$ which is larger as a suffix. Lemma 5.1 has already established that whenever $(u,v)$ is c-epichristoffel then $(uv,v)$ and $(u,uv)$ are also c-epichristoffel. Since epichristoffel words are words in the conjugacy class, they are c-epichristoffel. As they are the smallest in the conjugacy class, they are epichristoffel words. 
\end{remark}

\begin{theorem}
Every word appearing in the epichristoffel tree is obtained by the application of a finite number of episturmian morphisms to a letter and is lexicographically smallest in its respective conjugacy class making it an epichristoffel word. 
\end{theorem}
\begin{proof}
    The root of the epichristoffel tree is a word $w=w[1]w[2]...w[n]$ of length $n$ that is the smallest word in the conjugacy class of a word $w'$ obtained by an application of a sequence of episturmian morphisms as shown in the construction and is thus an epichristoffel word. \newline
    This root word $w$ is factorized as\\ $w[1]w[2]...w[|u|],w[|u|+1]w[|u|+2]...w[|u|+|v|]$ or $w[1]w[2]...w[|v|],w[|v|+1]w[v+2]...w[|v|+|u|]$ \\ where $|u|+|v|=n$, depending on whether $w[1]w[2]...w[|u|]$ is the epichristoffel word corresponding to the word $u$ or $w[1]w[2]...w[|v|]$ is the epichristoffel word corresponding to the word $v$ respectively, where $u,v$ is the factorisation of $w'$ as the product of two $c-$ epichristoffel words as shown in Example 3.3. \newline Without loss of generality, if $w[1]w[2]...w[|u|],w[|u|+1]w[|u|+2]...w[|u|+|v|]$ is the factorization of $w$ then the left and right descendants in the tree, respectively, are \\ $w[1]w[2]...w[|u|],w[1]w[2]...w[|u|]w[|u|+1]w[|u|+2]...w[|u|+|v|]$ and \\$w[1]w[2]...w[|u|]w[|u|+1]w[|u|+2]...w[|u|+|v|],w[|u|+1]w[|u|+2]...w[|u|+|v|]$. \newline
    Since $w[1]w[2]...w[n]$ is the smallest in its conjugacy class, we have \\ $w[1]w[2]...w[|u|]< w[|u|+1]w[|u|+2]...w[|u|+|v|]$. \newline Thus, the left and right descendants are also the smallest in their respective conjugacy classes and by Lemma 5.1 they are $c-$ epichristoffel due to which these descendants are also epichristoffel words. This pattern continues for all the words appearing in the epichristoffel tree and thus every word in this tree is an epichristoffel word.
    
\end{proof}
Thus, using an arbitrary epichristoffel word, the epichristoffel tree can be used to determine the existence of epichristoffel words of various lengths. Although a characterization of the epichristoffel word for each conjugacy class was not found in \cite{Onagen}, the epichristoffel tree can give us the characterization of an epichristoffel word corresponding to a tuple present in the tree. This is illustrated below:

\begin{example}\label{ex5.3}
 To determine the epichristoffel word in the conjugacy class of the $c$-epichristoffel word of the tuple, $(3,8,16)$ we need only to trace the path from the root of the epichristoffel tree to the node that represents the tuple (3,8,16). The root of this tree, as seen in Example \ref{ex5.1} is the word (xzyz,zyz). To find the word corresponding to the tuple (3,8,16), we must proceed to the right child given by the node, $(xzyzzyz,zyz)$ followed by its left child given by the node $(xzyzzyz,xzyzzyzzyz)$, which is followed by its right child given by the node $(xzyzzyzxzyzzyzzyz,xzyzzyzzyz)$. Thus, the epichristoffel for the conjugacy class of the epichristoffel tuple (3,8,16) is given by the word $xzyzzyzxzyzzyzzyzxzyzzyzzyz$.
\end{example}
\begin{remark}
    	We are using the epichristoffel tree to give a characterization of each epichristoffel word of its conjugacy class that appears on the tree as this was the question raised by Paquin in \cite{Onagen}. For example, the c-epichristoffel word corresponding to the tuple (3,8,16) is determined by applying the episturmian morphisms $\psi_{z} \psi_{y} \psi_{z} \psi_{x} \psi_{z} \psi_{z}$ to the letter x since, (3,8,16) $\to_{z}$ (3,8,5) $\to_{y}$ (3,0,5) $\to_{z}$ (3,0,2) $\to_{x}$ (1,0,2) $\to_{z}$ (1,01) $\to_{z}$ (1,0,0) . Doing so we obtain $\psi_{z}\psi_{y}\psi_{z}\psi_{x} \psi_{z} \psi_{z} (x)=\psi_{z}\psi_{y} \psi_{z} \psi_{x} \psi_{z} (zx)=\psi_{z} \psi_{y} \psi_{z} \psi_{x} (zzx)=\psi_{z} \psi_{y} \psi_{z} (xzxzx)=\psi_{z} \psi_{y} (zxzzxzzx)=\psi_{z} (yzyxyzyzyxyzyzyx)=zyzzyzxzyzzyzzyzxzyzzyzzyzx$.
The question raised by Paquin was to determine, given any c-epichristoffel word, such as the one above, a characterization of the epichristoffel word in its conjugacy class. We have shown that by locating the word on the epichristoffel tree, we receive the desired word in the conjugacy class. This has been explained in Example 5.3. Thus, by constructing an epichristoffel tree and its corresponding Stern-Brocot tree, we have successfully determined the epichristoffel word in the conjugacy class of each word on the tree. In Example 5.3, the tree gives us a characterization of the epichristoffel word in the conjugacy class of all the tuples occurring in it such as (2,3,6),(1,3,6),(3,4,8),(3,5,10),(2,5,10),(1,4,8),(3,8,16) etc.

\end{remark}

The following result is obtained by ordering the epichristoffel tree using the definition of left and right diagonals of Section 4.

\begin{theorem}
If $w$ is an epichristoffel word corresponding to the tuple $(x_1, x_2, \ldots x_k)$, then there exist epichristoffel words for tuples of the form $(x_1+d_1, x_2+d_2, \ldots x_k+d_k)$, for some positive integers $d_1,d_2,\dots,d_k$.
\end{theorem}
\begin{proof}
Since the mediant operation used for two fractions in the Stern-Brocot tree is extended for the case of $k$-tuples, Lemma \ref{lem4.2} can be applied analogously for the case of the epichristoffel tree.\\
Hence, each left diagonal of the Stern-Brocot tree corresponding to the epichristoffel tree is the sum of some previous left diagonals of the tree. \
Now, if $L$ and $L^\prime$ are any two left diagonals that satisfy the hypothesis then their sum also satisfies the hypothesis. Since \\
$L = \{(x_1,x_2,\ldots ,x_k), (x_1+d_1,x_2+d_2,\ldots x_k+d_k), (x_1+2d_1,x_2+2d_2,\ldots x_k+{2d}_k), \dots\}$ and \newline
$L^\prime = \{(y_1,y_2,\ldots,y_k), (y_1+e_1,y_2+e_2,\ldots y_k+e_k), (y_1+{2e}_1,y_2+2e_2,\ldots y_k+{2e}_k), \ldots\}$,\\then,\\
$L \oplus L^\prime = \{(x_1+y_1,x_2+y_2,\ldots,x_k+y_k), (x_1+y_1+d_1+e_1,x_2+y_2+d_2+e_2,\ldots,x_k+y_k+d_k+e_k), \ldots\}$\\
Thus, each left diagonal of the Stern-Brocot tree corresponding to an epichristoffel tree satisfies the result.\\
\end{proof}
\begin{remark}
    	The Stern-Brocot trees corresponding to the epichristoffel trees follows the same construction but for tuples rather than fractions. That is, $(a,b,c)\oplus(d,e,f)=(a+d,b+e,c+f)$. Therefore, the result can be extended analogously. 
\end{remark}
\begin{remark}
    Theorem 5.2, is an extension of Lemma 4.2 for the case of tuples, that is, every left diagonal is the mediant sum of some previous two left diagonals.  We are using this to show the existence of epichristoffel words of various lengths in the latter part of the paper. We use $L$ and $L'$ as an example to show that if they are any two left diagonals that satisfy the hypothesis then their mediant sum also satisfies the hypothesis. Since every left diagonal is the mediant sum of previous two left diagonals and if those two left diagonals satisfy the hypothesis then their sum must also satisfy the hypothesis. 
\end{remark}
The above result helps to answer the question raised by Paquin in \cite{Onagen} on the existence of epichristoffel words of various lengths. This is illustrated in the example below. 

\begin{example}\label{ex5.4}
The first three left diagonals of the epichristoffel tree in Example \ref{ex5.1} produce epichristoffel words of lengths $4n+3, 8n+2$ and $12n+5$ for $n = 1, 2, 3 \ldots$
\begin{gather*}
L_1 = \{(1,2,4), (2,3,6), \ldots ,(n,n+1,2n+2), \ldots\}\\
L_2 = \{(1,3,6), (3,5,10), \ldots ,(2n-1,2n+1,4n+2), \ldots\}\\
L_3 = \{(2,5,10), (5,8,16), \ldots, (3n-1,3n+2,6n+4), \ldots\}
\end{gather*}
The first three right diagonals of the epichristoffel tree in Example \ref{ex5.1} produce epichristoffel words of lengths $3n+4, 6n+5$ and $9n+9$ for $n = 1, 2, 3 \ldots$
\begin{gather*}
R_1 = \{(1,2,4), (1,3,6), \ldots ,(1,n+1,2n+2), \ldots\}\\
R_2 = \{(2,3,6), (2,5,10), \ldots, (2,2n+1,4n+2), \ldots\}\\
R_3 = \{(3,5,10), (3,8,16), \ldots ,(3,3n+2,6n+4), \ldots\}
\end{gather*}
\end{example}
\begin{remark}
    The computation of the left diagonals $L_1,L_2,L_3,\ldots$ and the right diagonals $R_1,R_2,R_3,\ldots$ follow from Definition 4.4. Consider the Stern-Brocot tree in Example 5.1. Here $L_1$ denotes the 1st term from each level beginning at the first level, that is, ${(1,2,4),(2,3,6),(3,4,8),\ldots}$. $L_2$ denotes the 2nd term from each level. Since the first level has only one term, we begin at the next level, that is, ${(1,3,6),(3,5,10),(5,7,14),\ldots}$ and so on. 
    $R_1$ denotes the 1st term taken from the end of each level, that is ${(1,2,4),(1,3,6),(1,4,8),\ldots}$.
$R_2$ denotes the 2nd term taken from the end of each level, that is, ${(2,3,6),(2,5,10),\ldots}$ and so on. 

\end{remark}
\begin{theorem}
For a 3-letter alphabet, $A = \{x,y,z\}$, there exist epichristoffel words having occurences of each letter at least once of every length $n \geq 4$ except for $n = 5$.
\end{theorem}

\begin{proof}
We first begin by observing that for any given even number $n \geq 4$,  one can determine an epichristoffel word by simply considering a tuple of the form $(1, 1, 2m)$ where $m = 1, 2, 3, \ldots$\\
It is easy to see that a tuple of the above form satisfies the condition of Proposition \ref{prop3.1}.\\
The right diagonal $R_1$ in Example \ref{ex5.4} shows the existence of epichristoffel words of length $3n+4$. Since we have already established the existence of words of even length, we can ignore the words of even length from words with length of the form $3n+4$ obtaining words with lengths $7, 13, 19, \ldots, 6n+1$,  for $n = 1, 2, 3, \ldots$\\
Constructing an epichristoffel tree with the root node as the word $(xzz, yzz)$ which corresponds the tuple (1, 1, 4), the left-diagonal $L_1$ yields tuples of the form $(n, 1, 2n+2)$ that guarantees the existence of epichristoffel words of length, $3n+3$ from which we can conclude that there exist epichristoffel words of length $6n+3$.\\
The right diagonal $R_2$ in Example \ref{ex5.4} yields epichristoffel words of length $6n+5$. Thus, the existence of epichristoffel words of lengths $6n+1, 6n+3, 6n+5$ are established and since any epichristoffel word of even length also exists, the result is established.
\end{proof}
Another way that Christoffel and epichristoffel words differ from one another is that, as seen in Lemma \ref{lem3.4}, epichristoffel words are not necessarily factorizable as a product of two epichristoffel words, in contrast to Christoffel words, which can always be factorized as a product of two Christoffel words. The epichristoffel tree can be used to identify some subclasses of epichristoffel words that can or cannot be factorized as a product of epichristoffel words, as we demonstrate below.
 
To see this, we start with an arbitrary epichristoffel word that cannot be factorized as the product of two other epichristoffel words. For instance, the epichristoffel word $xzyzzyz$, which appears as the root of the epichristoffel tree with the factorization $(xzyz,zyz)$, corresponding to the tuple (1,2,4). Here, the word $zyz$ is not epichristoffel. Consider the path that starts at the tree's root and descends to its right child, namely $(xzyzzyz,zyz)$. Concatenating the word $zyz$ to the word $xzyzzyz$ ensures that the word is still not factorizable as two epichristoffel words since $zyz$ is not an epichristoffel word. In a similar vein, the epichristoffel word $(xzyzzyz,zyz)$, when we take into consideration has its right child to be $(xzyzzyzzyz,zyz)$, which likewise cannot be factorized as the product of two epichristoffel words.
As a result, a subclass of epichristoffel words that lack the ability to be factorized as the product of two epichristoffel words can be identified with the aid of the epichristoffel tree. 

Continuing with the above example it is seen that the set:

$\{(1,2,4), (1,3,6), (1,4,8), \ldots (1,n+1,2n+2), \ldots\}$ has examples of tuples that cannot yield the product of two epichristoffel words. These tuples are precisely the right diagonal $R_1$ of the epichristoffel tree in Figure 5.

On the other hand, if we begin with a word that can be factorized as the product of two epichristoffel words, then all its descendants on the epichristoffel tree will possess a factorization as two epichristoffel words since if $(u,v)$ is such a node, with both $u$ and $v$ as epichristoffel, then $(u,uv)$ and $(uv,v)$ are factorizations of epichristoffel words as a product of epichristoffel words. We thus get the following results:

\begin{theorem}
If $w$ is an epichristoffel word that cannot be factorized as the product of two epichristoffel words then the right diagonal $R_{1}$ of the epichristoffel tree formed using $w$ as the root contains epichristoffel words that cannot be factorized as the product of two epichristoffel words. 
\end{theorem}
\begin{proof}
    Since $w$ cannot be factorized as the product of two epichristoffel words, constructing the \\ epichristoffel word using $w$ gives us the root node of the form $(u,v)$ where $v$ is not an epichristoffel word. \newline
    The right diagonal $R_{1}$ will then consist of words of the form $(u,v),(uv,v),(uvv,v)...$ and thus provides us with an infinite subclass of epichristoffel words that cannot be factorized as a product of two epichristoffel words.
\end{proof}

\begin{theorem}
If $w$ is an epichristoffel word that can be factorized as the product of two epichristoffel words then every word of the epichristoffel tree formed using $w$ as the root can be factorized as the product of two epichristoffel words. 
\end{theorem}
\begin{proof}
    If $w$ is an epichristoffel word that can be factorized as the product of two epichristoffel words then it appears in the epichristoffel tree in the form $(u,v)$ where $u$ and $v$ are both epichristoffel words. \newline
    The left and right descendants respectively are $(u,uv)$ and $(uv,v)$ and thus possess a factorization as the product of two epichristoffel words.
\end{proof}
\begin{remark}
    A related question is the following: Given a tuple, is there a unique c-epichristoffel word associated to it? Proposition 3.1 gives us the method to determine for a given tuple, if there exists a c-epichristoffel word associated to it. This word is unique up to conjugation or circular shifts as the method of applying the iterations may vary. For example, the tuple (1,2,4) is a tuple corresponding to a c-epichristoffel word. The c-epichristoffel word can be obtained by applying the episturmian morphisms $\psi_{z} \psi_{y} \psi_{z}$ to the letter $x$ or the episturmian morphisms $\psi_{z} \psi_{y} \psi_{x}$ to the letter $z$. The first case yields the word $zyzzyzx$, the second case yields the word $zyzxzyz$. Although they are different, they are both conjugate to each other, that is, one can be obtained from the other by circular shifts of the letters.
\end{remark}
\section{Conclusion}

In this paper, we have continued Paquin’s study in \cite{Onagen} on epichristoffel words by extending the definition of the Christoffel tree to accommodate epichristoffel words. In doing so, we can answer some of the questions raised, such as the existence of epichristoffel words of various lengths and provide a characterization for the conjugacy class of every word appearing on the tree. The epichristoffel tree can be a useful tool in determining epichristoffel words that possess a factorization property as smaller epichristoffel words and those that do not. Further studying these trees may help provide insights into the Fraenkel conjecture, which states that there exists a unique infinite word up to letter permutation and conjugation that is balanced and has pair-wise distinct letter frequencies over a $k$-letter alphabet, which is periodic and of the form $ppp \ldots$ where $p$ is an epichristoffel word.

\nocite{*}
\bibliographystyle{eptcs}
\bibliography{generic}
\end{document}